\documentclass[oribibl,envcountsame]{llncs}

\usepackage[
  hidelinks,
  bookmarks,
  bookmarksdepth=10,
  bookmarksopen,
  pdfpagemode=UseNone,
  unicode,
  pdftitle={Partiality, Revisited: The Partiality Monad as a Quotient Inductive-Inductive Type},
  pdfauthor={Thorsten Altenkirch, Nils Anders Danielsson and Nicolai Kraus}
  ]{hyperref}
\newcommand{\doi}[1]{doi:\href{https://doi.org/#1}{%
    \urlstyle{same}\nolinkurl{#1}}}

\usepackage{natbib}
\setcitestyle{numbers,square,semicolon,aysep={},yysep={,}}

\usepackage[utf8]{inputenc}
\usepackage[T1]{fontenc}
\DeclareUnicodeCharacter{03C9}{\ensuremath{\omega}}
\DeclareUnicodeCharacter{2291}{\ensuremath{\sqsubseteq}}

\usepackage{fancyhdr}
\usepackage{csquotes}
\usepackage{mathalfa}
\usepackage{amsmath}
\usepackage{amssymb}
\usepackage{amsfonts}
\usepackage{listings}
\usepackage{mathtools}
\usepackage{multicol}
\usepackage[llparenthesis,rrparenthesis]{stmaryrd}
\usepackage{xspace}
\usepackage{graphics}
\usepackage{braket}
\usepackage{float}
\usepackage{enumitem}
\usepackage{mathpartir}

\newcommand{\pbot}[1]{#1_{\bot}}
\newcommand{\order}{\sqsubseteq}
\newcommand{\LUB}{\bigsqcup}
\newcommand{\now}{\mathsf{now}}
\newcommand{\later}{\mathsf{later}}
\renewcommand{\Set}{\mathsf{Set}}
\newcommand{\Prop}{\mathsf{Prop}}
\newcommand{\Seq}{\mathsf{Seq}}
\newcommand{\ismon}{\mathsf{ismon}}
\newcommand{\DDname}{\mathrm D}
\newcommand{\DD}[1]{\DDname \! \left(#1\right)}
\newcommand{\quotient}[2]{{#1}/{#2}}
\newcommand{\DDinSubscript}[1]{\DDname}
\newcommand{\omegaCPOCategory}{\ensuremath{\mathsf{\omega\textsf{-}CPO}}}
\newcommand{\shift}{\mathit{shift}}
\newcommand{\unshift}{\mathit{unshift}}
\newcommand{\fstAlgmorph}{\mathsf{fst}}
\newcommand{\search}{\mathit{search}}
\newcommand{\caseta}{\left[\eta \mid \bot\right]}

\newcommand{\prd}[1]{\Pi_{#1} \, }
\newcommand{\lam}[1]{\lambda #1 .}
\newcommand{\sm}[1]{\Sigma_{#1} \, }
\newcommand{\Univ}{\mathcal{U}}
\newcommand{\N}{\mathbb{N}}
\newcommand{\Q}{\mathbb{Q}}
\newcommand{\zero}{\mathbf{0}}
\newcommand{\one}{\mathbf{1}}
\newcommand{\two}{\mathbf{2}}
\newcommand{\true}{1_{\two}}
\newcommand{\false}{0_{\two}}
\newcommand{\inl}{\mathsf{inl}}
\newcommand{\inr}{\mathsf{inr}}
\newcommand{\jdgeq}{\vcentcolon\equiv}
\DeclarePairedDelimiter\truncation{\lVert}{\rVert}
\newcommand{\trunc}[1]{\truncation*{#1}}

\begin{document}

\mainmatter

\title{Partiality, Revisited}
\subtitle{The Partiality Monad as a\\Quotient Inductive-Inductive Type}

\author{%
Thorsten Altenkirch\inst{1}\thanks{Supported by EPSRC grant EP/M016994/1 and by USAF, Airforce office for scientific research, award FA9550-16-1-0029.} 
\and Nils Anders Danielsson\inst{2}\thanks{Supported by a grant from the
  Swedish Research Council (621-2013-4879).}
\and Nicolai Kraus\inst{1}\thanks{Supported by EPSRC grant EP/M016994/1.}
}

\institute{University of Nottingham \and University of Gothenburg}

\maketitle

\fancyhf{}
\renewcommand{\headrulewidth}{0pt}
\fancyfoot[L]{\parbox{\textwidth}{\footnotesize{\textit{This is the
        authors' version of a paper published in FoSSaCS 2017. The
        final publication is available at Springer via}
      \url{http://dx.doi.org/10.1007/978-3-662-54458-7_31}\textit{.
        This version contains some insignificant improvements to the
        text that are not present in the published version.}}}}
\thispagestyle{fancy}

\begin{abstract}
  \citeauthor{capretta:2005}'s delay monad can be used to model
  partial computations, but it has the ``wrong'' notion of built-in
  equality, strong bisimilarity.
  An alternative is to quotient the delay monad by the ``right''
  notion of equality, weak bisimilarity.
  However, recent work by \citeauthor{Chapman2015} suggests that it is
  impossible to define a monad structure on the resulting construction
  in common forms of type theory without assuming (instances of) the
  axiom of countable choice.

  Using an idea from homotopy type theory---a higher
  inductive-inductive type---we construct a partiality monad without
  relying on countable choice.
  We prove that, in the presence of countable choice, our
  partiality monad is equivalent to the delay monad quotiented by weak
  bisimilarity.
  Furthermore we outline several applications.
\end{abstract}

\section{Introduction}

Computational effects can be modelled using monads, and in some
functional programming languages (notably Haskell) they are commonly
used as a program structuring device.
In the presence of dependent types one can both write and
reason about monadic programs.
From a type theorist's point of view, even a ``pure'' functional
language like Haskell is not really pure as it has built-in effects,
one of which is partiality: a function does not necessarily terminate.
It is thus natural to look for a partiality monad which makes it
possible to model partial computations and to reason about possibly
non-terminating programs.

\citeauthor{capretta:2005} modeled partial computations using a
coinductive construction that we call the \emph{delay
  monad}~\cite{capretta:2005}.
We use the notation $\DD A$ for \citeauthor{capretta:2005}'s type of
delayed computations over a type $A$.
$\DD A$ is coinductively generated by $\now : A \to \DD A$ and $\later : \DD A \to \DD A$.
Examples of elements of $\DD A$ include $\now(a)$ and
$\later(\later(\now(a)))$, as well as the infinitely delayed value
$\bot$, defined by the guarded equation $\bot = \later(\bot)$.
We can model recursive programs as Kleisli arrows $A \to \DD B$,
and we can construct fixpoints of ($\omega$-continuous) functions of
type $(A \to \DD B) \to (A \to \DD B)$, see
\citeauthor{benton-kennedy-varming}~\cite{benton-kennedy-varming}.

Unfortunately, \citeauthor{capretta:2005}'s delay monad is sometimes
too intensional.
It is often appropriate to treat two computations as equal if they
terminate with the same value, but the delay monad allows us to count
the number of ``steps'' ($\later$ constructors) used by a computation.

\citeauthor{capretta:2005} addressed this problem by defining a
relation that we call \emph{weak bisimilarity},
$\sim_{\DDinSubscript{A}}$, and that relates expressions that only
differ by a finite number of $\later$
constructors~\cite{capretta:2005}.
\citeauthor{capretta:2005} proved that the delay monad combined with
weak bisimilarity is a monad in the category of setoids.

A setoid is a pair consisting of a type and an equivalence relation
on that type.
Setoids are sometimes used to approximate quotient types in type
theories that lack support for quotients.
However, a major difference between setoids and quotient types is that
setoids do not provide a mechanism for information hiding.
Using the setoid approach basically boils down to introducing a new
relation together with the convention that all constructions have to
respect this relation.
A problem with this approach is that it can lead to something which
has informally become known as \emph{setoid hell}, in which one is
forced to prove that a number of constructions---even some that do not
depend on implementation details by, say, pattern matching on the
$\now$ and $\later$ constructors---preserve setoid relations.
This kind of problem does not afflict quotient types.

In a type theory with quotient types~\cite{hofmann:thesis}, one can
consider using the quotient $\quotient{\DD A}{\sim_{\DDinSubscript{A}}}$ as the type of
partial computations of type $A$.
This idea was discussed in a talk by Uustalu~\cite{tarmo-slides},
reporting on joint work with \citeauthor{capretta:2005} and the
first-named author of the current paper.
However, the idea does not seem to work as intended.
It is an open problem---and believed to be impossible---to show that
this construction actually constitutes a monad (in ``usual'' forms of
type theory).

With an additional assumption, \citeauthor{Chapman2015} have managed
to show that the partiality operator $\quotient{\DD
  -}{\sim_{\DDinSubscript{-}}}$ is a monad~\cite{Chapman2015}.
This additional assumption is known as \emph{countable choice}.
To express what this is, first note that the \emph{propositional truncation}, written $\trunc -$ and sometimes called ``squashing'', is an operation that turns a type into a proposition (a type with at most one element). 
We can see $\trunc A$ as the quotient of $A$ by the total relation.
Countable choice says that $\Pi$ and $\trunc -$ commute if the domain is the natural numbers,
in the sense that there is a function from
$\prd{n:\N} \trunc{P(n)}$
to $\trunc{\prd{n:\N} P(n)}$.
Even though this principle holds in some models, its status in type
theory is unclear: the principle is believed to be independent of
several variants of type theory.
Recently \citeauthor{coquand_mannaa_ruch} have shown that it cannot be
derived in a theory with propositional truncation and a single
univalent universe~\cite{coquand_mannaa_ruch},
speculating that the result might extend to a theory with a hierarchy
of universes.
Furthermore \citeauthor{richman2001constructive} argues that countable
choice should be avoided in constructive
reasoning~\cite{richman2001constructive}.
The main purpose of the present paper is to define a partiality monad
without making use of this principle.

The situation with the quotiented delay monad is similar to that of
one variant of the real numbers in constructive mathematics.
If the Cauchy reals are defined as a quotient, then it is impossible
to prove a specific form of the statement that every Cauchy sequence
of Cauchy reals has a limit using $\mathsf{IZF}_\mathrm{Ref}$, a
constructive set theory without countable choice~\cite{lubarsky}.
It is suspected that corresponding statements are also impossible to
prove in several variants of type theory.
An alternative solution was put forward in the context of homotopy
type theory~\cite{HoTTbook}.
In that approach, the reals are constructed inductively simultaneously
with a notion of closeness, and the quotienting is done directly in
the definition using a \emph{higher inductive-inductive type} (HIIT).

In 2015, Andrej Bauer and the first-named author of the current paper
suggested to use a similar approach to define a partiality monad
without using countable choice, an idea which was mentioned by
\citeauthor{Chapman2015}~\cite{Chapman2015}. Here, we show that this
is indeed possible.

We do not make use of the full power of HIITs, but restrict ourselves
to \emph{set-truncated} HIITs.
We call such types \emph{quotient
  inductive-inductive types}, \emph{QIITs}, following
\citeauthor{ttintt}~\cite{ttintt}.
Some of the theory of QIITs is developed in the forthcoming PhD thesis
of \citeauthor{gabe:thesis}~\cite{gabe:thesis}, see also
\citeauthor{altenkirch2016quotient}~\cite{altenkirch2016quotient}.
Although type theory extended with QIITs is still experimental and
currently lacks a solid foundation, QIITs are a significantly simpler
concept than full-blown HIITs.
It is conjectured that QIITs exist in some computational models of
type theory.

The type theory that we work in can be described as a fragment of the
theory considered in the standard textbook on homotopy type
theory~\cite{HoTTbook} (henceforth referred to as the HoTT book), and
is quite close to the theory considered by
\citeauthor{Chapman2015}~\cite{Chapman2015}.
Details are given in Sect.~\ref{sec:background}.
The construction of our partiality monad is given in
Sect.~\ref{sec:partiality}, together with its elimination principle
and some properties.
Furthermore we show that it gives us \emph{free $\omega$-cpos} in
a sense that we will make precise.
In Sect.~\ref{sec:countchoice-equiv} we show that, assuming countable
choice, our partiality monad is equivalent to (in bijective
correspondence to) the one of
\citeauthor{Chapman2015}~\cite{Chapman2015}.
We outline some applications of the partiality monad in Sect.~\ref{sec:applications}, and conclude with a short discussion in Sect.~\ref{sec:discussions}.

\subsubsection*{Agda Formalisation.}

The paper is accompanied by a formal
development~\cite{altenkirch-danielsson-kraus-partiality-agda} in
Agda.

At the time of writing, Agda does not directly support QIITs.
We have chosen to represent them by postulating their elimination principles together with the equalities they are supposed to satisfy.
In some cases (but not for the partiality monad) we have also made use of Agda's experimental rewriting feature~\cite{cockx-abel-rewriting} to turn postulated equalities into judgmental computation rules.

Note that there are differences between the Agda code and the
presentation in the text.
For one, the formalisation discusses various additional topics that
have been omitted in the paper for reasons of space, and is more
rigorous.
Furthermore, the paper defines the partiality monad's elimination
principle as a universal property.
In the formalisation the elimination principle is given as an
induction principle, but we also prove that this principle is
interderivable with the universal property.
Finally there are a number of small differences between the
formalisation and the text, and some results in the paper have not
been formalised at all, most notably the results about the reals in
Sect.~\ref{subsec:reals}.

\section{Background: Type Theory with Quotient Inductive-Inductive Types} \label{sec:background}

We work in intensional type theory of Martin-Löf style with all the
usual components (e.g.\@ $\Pi$, $\Sigma$, inductive types), including
the identity type (we use the notation $x = y$).
We assume that equality of functions is extensional, and that
(strong) bisimilarity implies equality for coinductive types.

\citeauthor{Chapman2015}~\cite{Chapman2015} assume the axiom of
\emph{uniqueness of identity proofs}, UIP, for all small types (types
in the lowest universe).
UIP holds for a type $A$ if, for any elements $x$, $y : A$, if we have
equalities $p$, $q : x = y$, then we have $p = q$.
Instead of postulating an axiom, we prefer to work in a more general
setting and restrict
ourselves to types with the corresponding property. This approach is
compatible with homotopy type theory.
In the language of homotopy type theory, we work with \emph{sets} or
\emph{$0$-truncated types}; a type is a set if and only if it
satisfies UIP.
When we write $A : \Set$, we mean that $A$ is a type (in some
universe) with the property of being a set; and when we write $B : A
\to \Set$, we mean that $B$ is a family of types such that each $B(a)$
is a set.

Similarly to $A : \Set$, we write $P : \Prop$ for a type $P$ with the
property that it is a proposition, i.e.\@ a $(-1)$-truncated type,
i.e.\@ a type with the property that any two of its elements are
equal.
A proposition is also a set.
The type of all propositions in a certain universe is closed under all
operations that are relevant to us, and the same applies to sets.

In addition to UIP, \citeauthor{Chapman2015}~\cite{Chapman2015} assume
\emph{propositional extensionality}---that logically equivalent
propositions are equal---for all small propositions.
This property is equivalent to the \emph{univalence
  axiom}~\cite{HoTTbook}, restricted to (small) propositions.
Just like \citeauthor{Chapman2015}, we only require propositional
extensionality (not full univalence) for our development, with the
exception that univalence is used to show that certain precategories
(in the sense of the HoTT book~\cite{HoTTbook}) are categories.
For an example of how propositional extensionality is used, see
Lemma~\ref{lem:order-lemma}.

\citeauthor{Chapman2015}~\cite{Chapman2015} also assume the existence
of quotient types in the style of
\citeauthor{hofmann:thesis}~\cite{hofmann:thesis}.
Given a set $A$ and a propositional relation $\sim$ on it, the (set-)
quotient $\quotient{A}{\sim}$ can in homotopy type theory be constructed as
a higher inductive type with three constructors~\cite{HoTTbook}:
\begin{alignat*}{2}
  &[-] && : A \to \quotient{A}{\sim} \\
  &[-]^= && :  \prd{a,b:A} a \sim b \to [ a ] = [ b ] \\
   &\mathsf{irr} && : \prd{x,y:\quotient{A}{\sim}}\prd{p,q : x = y}  p = q
\end{alignat*}
The last constructor $\mathsf{irr}$ ensures that any two parallel equalities are equal, that is, that $\quotient{A}{\sim}$ is set-truncated.
We call a higher inductive type with such a set-truncation constructor a \emph{quotient inductive type} (QIT). 

As noted above, \citeauthor{Chapman2015}~\cite{Chapman2015} use
countable choice, which we want to avoid.
Instead we make use of quotient \emph{inductive-inductive} types
(QIITs)~\cite{gabe:thesis,altenkirch2016quotient}.
From the point of view of homotopy type theory, these are
set-truncated higher inductive-inductive types (HIITs); some other
examples of HIITs can be found in the HoTT
book~\cite[Chap.\@~11]{HoTTbook}.
While it seems plausible that QIITs exist in some computational
models of type theory, this has yet to be determined.

\section{The Partiality Monad} \label{sec:partiality}

As indicated in the introduction we define the partiality monad $\pbot{(-)}$ as a QIIT,
defining the type $\pbot A$ simultaneously with an ordering relation
$\order$ on $\pbot A$.
We will first describe the constructors and the elimination principle
of this definition, and then show that $\pbot A$ is the underlying type of the free
$\omega$-cpo (see Definition~\ref{def:omega-cpo}) on $A$, thus
proving that $\pbot{(-)}$ is a monad.
In the final part of this section we will give a characterisation of
the ordering relation; this is perhaps not as trivial as one might
expect, given the relation's definition.

Note that our construction of $\pbot A$ can be seen as a further example of a free algebraic structure defined in type theory.
It was discussed in the HoTT book~\cite[Chap.\@~6.11]{HoTTbook} that free groups can be defined as (in our terminology) quotient inductive types, while it is well-known that even simpler examples can be defined as ordinary inductive types.

\subsection{The Definition and Its Elimination Principles}

Let $A$ be a set.
We define the set $\pbot A$ simultaneously with a binary propositional
relation on $\pbot A$, written $\order$.
The set $\pbot A$ is generated by the following four
constructors, plus a set-truncation constructor:
\begin{alignat*}{6}
 &\eta &\;& : &\;& A \to \pbot A &\qquad\qquad& \LUB &\;& : &\;& \left(\sm{s : \N \to \pbot A} \prd{n:\N} s_n \order s_{n+1}\right) \to \pbot A \\
 &\bot && : && \pbot A && \alpha && : && \prd{x,y : \pbot A} x \order y \to y \order x \to x = y
\end{alignat*}
The constructor $\eta$ tells us that any element of $A$ can be viewed
as an element of $\pbot A$, and $\bot$ represents a non-terminating
computation.
The constructor $\LUB$ is intended to form least upper bounds of
increasing sequences, and $\alpha$ ensures that the ordering relation
$\order$ is antisymmetric.

The relation $\order$ is a type family that is indexed twice by
$\pbot A$.
It is generated by six constructors.
One of these constructors says that, for any $x, y : \pbot A$, the
type $x \order y$ is a proposition
(\mbox{$\prd{p,q : x \order y}  p = q$}).
Because any two proofs of $x \order y$ are equal, we do not name the
constructors of the ordering relation.
The remaining constructors are given as inference rules, where each
rule is implicitly $\Pi$-quantified over its unbound variables (the
same comment applies to other definitions below):
\begin{mathpar}
  \inferrule{ }{x \order x}
  \and
  \inferrule{x \order y \and y \order z}{x \order z}
  \and
  \inferrule{ }{\bot \order x}
  \and
  \inferrule{ }{\prd{n:\N}s_n \order \LUB(s,p)}
  \and
  \inferrule{\prd{n:\N}s_n \order x}{\LUB(s,p) \order x}
\end{mathpar}
The rules state that $\order$ is reflexive and transitive, that $\bot$ is at least as small as any other element of $\pbot A$, and that $\LUB$ constructs least upper bounds.

Now we will give the elimination principle of $\pbot {(-)}$ and
$\order$.
This principle can be stated in different ways.
One way would be to state it as an induction principle, along the
following lines:
\emph{
 Given a family
 $P : \pbot A \to \mathsf{Set}$ and [\ldots something
 for $\order$\ldots], and given elements of $P(\bot)$, $\prd{a :
   A}P(\eta(a))$, [\ldots and so on\ldots], we can conclude that
 $\prd{x : \pbot A} P(x)$ and [\ldots].
}
We take this approach in our formalisation; for another example, see
the presentation of the Cauchy reals in the HoTT
book~\cite[Chap.\@~11.3.2]{HoTTbook}.
However, because the two types are defined simultaneously and involve
constructors targeting the equality type, the induction principle may
look somewhat involved and perhaps even ad-hoc, and it may not be
obvious that it is the ``correct'' one.

Instead, we present a universal property.
\citeauthor{gabe:thesis}~\cite{gabe:thesis} and \citeauthor{altenkirch2016quotient}~\cite{altenkirch2016quotient} have worked out a general form and rules for a large class of quotient
inductive-inductive types.
In their setting, any specification of a QIIT gives rise to a
\emph{category of algebras}, following methods that have been used for
W-types~\cite{awodeyGamSoja_indTypesInHTT} and certain higher
inductive types~\cite{Sojakova14}, and if this category has a
(homotopy-) initial object, then this object is taken as the
definition of the QIIT.
We use the following algebras:
\begin{definition}[partiality algebras] \label{def:part-alg}
 A \emph{partiality algebra} over the set $A$ consists of 
 a set $X$;
 a propositional binary relation on $X$, $\order_X$;
  an element $\bot_X : X$, a family $\eta_X : A \to X$, and a family $\LUB_X : (\sm{s : \N \to X}\prd{n : \N} s_n \order_X s_{n+1}) \to X$;
  and the following laws:
\begin{equation*}
 \begin{alignedat}{6}
 & x \order_X x & \qquad  & x \order_X y \to y \order_X z \to x \order_X z \\
 & \bot_X \order_X x && x \order_X y \to y \order_X x \to x = y \\
 & \prd{n : \N} s_n \order_X \LUB_X(s,p) &&
   \left(\prd{n : \N} s_n \order_X x\right) \to \LUB_X(s,p) \order_X x
 \end{alignedat}
\end{equation*}
The type $X$ and the type family $\order_X$ are allowed to target
universes distinct from the one that $A$ lives in.

For two partiality algebras over the same set $A$,
$(X, \order_X, \bot_X, \eta_X, \LUB_X)$ and
$(Z, \order_Z, \bot_Z, \eta_Z, \LUB_Z)$, a \emph{morphism of partiality algebras} from the former to the latter consists of a function ${f : X \to Z}$ satisfying the
following laws:
First, $f$ has to respect the ordering relation, $f^{\order} : x \order_X y \to f(x) \order_Z f(y)$.
Second, $f$ has to preserve some of the constructors, $f(\bot_X) = \bot_Z$,
$f \circ \eta_X = \eta_Z$, and $f(\LUB_X(s,p)) = \LUB_Z(f \circ s, f^{\order} \circ p)$.

Let us denote this structure of objects and morphisms by $\mathsf{Part}_A$.
\end{definition}
The structure $\mathsf{Part}_A$ is a category, in which the identity
morphism is the identity function, and composition of morphisms is
composition of functions.

We can now make the elimination principle precise.
Note that the tuple $(\pbot A, \order, \bot, \eta, \LUB)$ is a
partiality algebra.
As the \emph{elimination principle of $\pbot A$ and $\order$} we take
the statement that there is a unique (up to equality) morphism from
this partiality algebra to any other partiality algebra over $A$.
In the terminology of the HoTT book~\cite{HoTTbook},
the statement that there is a morphism is basically the recursion
principle of $\pbot{(-)}$ and $\order$, while uniqueness gives us the
power of the induction principle (with \emph{propositional}
computation rules).
Note that allowing the type $X$ and the type family $\order_X$ to
target arbitrary universes enables us to make use of large
elimination.

We do not lose anything by using a universal property instead of an
induction principle, at least for the induction principle referred to
in the following theorem. The theorem is similar to results due
to \citeauthor{gabe:thesis}~\cite{gabe:thesis} and
\citeauthor{altenkirch2016quotient}~\cite{altenkirch2016quotient}. It
is stated without proof here, but a full proof of the fact can be
found in our Agda development.
\begin{theorem}
  The elimination principle of $\pbot A$ and $\order$ can be stated as
  an induction principle.
  This induction principle, which comes with propositional rather than
  definitional computation rules, is interderivable with the universal
  property given above.
\end{theorem}
  Note that we could have defined $\pbot A$ and $\order$ differently.
  For instance, we could have omitted the set-truncation constructor
  from the definition of $\pbot A$, and then proved that the type is a
  set, following the approach taken for the Cauchy reals in the
  HoTT book~\cite{HoTTbook}.
  However, if we had done this, then our definitions would have been
  less close to the general framework mentioned
  above~\cite{gabe:thesis,altenkirch2016quotient}.

As a simple demonstration of the universal property we construct an
induction principle for $\pbot A$ that can be used when eliminating
into a proposition.
Following the terminology of the HoTT
book~\cite[Chap.\@~11.3.2]{HoTTbook}, we call it \emph{partiality
  induction}:
\begin{lemma}
  \label{lem:simplified-elim}
  Let $P$ be a family of propositions on $\pbot A$ such that both $P(\bot)$
  and ${\prd{a:A} P(\eta(a))}$ hold.
  Assume further that, for any increasing sequence
  $s :$ \mbox{$\N \to \pbot A$} (with corresponding proof $p$),
  $\prd{n : \N} P(s_n)$ implies $P(\LUB(s,p))$.
  Then we can conclude $\prd{x : \pbot A} P(x)$.
\end{lemma}
\begin{proof}
  The proof uses a standard method.
  We define a partiality algebra where the set is $Z \jdgeq \sm{x :
    \pbot A} P(x)$; the binary relation is $\order$, ignoring the
  second projections of the values in $Z$; and the rest of the algebra
  is constructed using the assumptions.
  The universal property gives us a morphism $m$ from the initial
  partiality algebra to this one, and in particular a function of type
  $\pbot A \to Z$.
  We are done if we can show that this function, composed with the
  first projection, is the identity on $\pbot A$.
  Note that the first projection can be turned into a partiality
  algebra morphism $\fstAlgmorph$.
  Thus, by uniqueness, the composition of $\fstAlgmorph$ and $m$ has to
  be the unique morphism from the initial partiality algebra to
  itself, and the function component of this morphism is the identity.
  \qed
\end{proof}

\subsection{ω-Complete Partial Orders}
\label{sec:omega-cpos}

Another way of characterising our quotient inductive-inductive
partiality monad is to say that $\pbot A$ is the \emph{free (pointed)
  $\omega$-cpo} over $A$:
\begin{definition}
  \label{def:omega-cpo}
  Let us denote the category $\mathsf{Part}_{\zero}$, where
  $\zero$ is the empty type, by \omegaCPOCategory{}. An
  \emph{$\omega$-cpo} is an object of this category.
\end{definition}
Let us quickly check that this definition makes sense.
A partiality algebra on $\zero$ is a set $X$ with a binary
propositional relation $\order_X$ that is a partial order.
There is a least element $\bot_X$ and any increasing sequence has a
least upper bound.
There is also a function of type $\zero \to X$, which we omit below as it carries no information.

We can now relate the category of
sets~\cite[Example~9.1.7]{HoTTbook}, written $\mathsf{SET}$, to the
category \omegaCPOCategory{}.
For any $\omega$-cpo we can take the underlying set, and it is easy to
see that this yields a functor, in the sense of the HoTT
book~\cite[Definition~9.2.1]{HoTTbook}, $\mathsf{U} :
\omegaCPOCategory{} \to \mathsf{SET}$.

We also have a functor $\mathsf{F} : \mathsf{SET} \to \omegaCPOCategory{}$, constructed as follows:
The functor maps a set $A$ to the $\omega$-cpo $(\pbot A, \order, \bot, \LUB)$.
For the morphism part, assume that we have a function $f : A \to B$.
Then $(\pbot B, \order, \bot, \eta \circ f, \LUB)$ is an $A$-partiality algebra, and hence there is a morphism to this algebra from the initial $A$-partiality algebra $(\pbot A, \order, \bot, \eta, \LUB)$.
By removing the components $\eta \circ f : A \to \pbot B$ and $\eta : A \to \pbot A$ we get a morphism between $\omega$-cpos.

The function $\eta$ lifts to a natural transformation from the
identity functor to $\mathsf{U} \circ \mathsf{F}$.
In order to construct a natural transformation from $\mathsf{F} \circ
\mathsf{U}$ to the identity functor, assume that we are given some
$\omega$-cpo $X$.
We can construct an $\omega$-cpo morphism from
$\mathsf{F}(\mathsf{U}(X))$ to $X$ by noticing that $(\mathsf{U}(X),
\order_X, \bot_X, \mathit{id}, \LUB_X)$ is a partiality algebra on
$\mathsf{U}(X)$, and thanks to initiality we get a morphism
$m$ from $\mathsf{F}(\mathsf{U}(X))$ to $X$
satisfying $m \circ \eta = \mathit{id}$.
After proving some equalities we end up with the following result,
where the definition of ``adjoint'' is taken from the HoTT
book~\cite[Definition~9.3.1]{HoTTbook}:
\begin{theorem}
  \label{thm:left-adjoint}
 For a given set $A$, the functor $\mathsf F$ is a left adjoint to the forgetful functor $\mathsf U$.
 This means that $\mathsf F(A)$ can be seen as the \emph{free} $\omega$-cpo over $A$. \qed
\end{theorem}
Thus we get a justification for calling the concept that we are
discussing the partiality \emph{monad}:
\begin{corollary}
  The composition
  $\mathsf U \circ \mathsf F : \mathsf{SET} \to \mathsf{SET}$, which
  maps objects $A$ to $\pbot A$, is a monad. \qed
\end{corollary}

 Note that one can also construct a monad structure on $\pbot {(-)}$ directly.
 Let us fix the set $A$.
 The unit is given by $\eta$.
 For the multiplication $\mu : \pbot{(\pbot A)} \to \pbot A$, note
 that $\pbot A$ can be given the structure of a partiality algebra
 over $\pbot A$ in a trivial way:
 the underlying set is $\pbot A$, the function
 $\eta_{\pbot A} : \pbot A \to \pbot A$ is the identity,
 $\order_{\pbot A}$ is $\order$, and so on.
 This gives us the function $\mu$ as the unique morphism from the
 initial partiality algebra to this one.
 Proving the monad laws is straightforward.

\subsection{A Characterisation of the Relation ⊑} \label{subsec:characterisation-of-order}

To further analyse the QIIT construction, we show how the relation
$\order$ on the set $\pbot A$ behaves.\footnote{The work presented in
  Sect.~\ref{subsec:characterisation-of-order} was done in
  collaboration with Paolo Capriotti.}
These results are useful when working with the partiality monad, and
will play an important role in the next section of the paper.
The arguments are only sketched here, details are given in the formalisation.
We use the propositional truncation $\trunc -$ (also known as
``squashing''), which turns a type into a proposition. It can be
implemented by quotienting with the trivial relation.

We know that $\bot \order y$ is (by definition) satisfied for any $y :
\pbot A$, and for the least upper bound we have that $\LUB(s,q) \order
y$ is equivalent to $\prd{n:\N} s_n \order y$. The following lemma
provides a characterisation of $\eta(a) \order y$, for any $a : A$:

\begin{lemma}\label{lem:order-lemma}
The binary relation $\order$ on $\pbot A$ has the following properties:
\begin{equation*}
 \begin{alignedat}{3}
  \eta(a) &\order \bot &\quad & \leftrightarrow &\quad & \mathsf{0} \\
  \eta(a) &\order \eta(b) &\quad & \leftrightarrow &\quad & a = b \\
  \eta(a) &\order \LUB(s,q) && \leftrightarrow && \trunc{\sm{n:\N}\eta(a) \order s_n}
 \end{alignedat}
\end{equation*}
\end{lemma}
We will give the proof of this lemma later and make a remark first.
Constructors in ``HIT-like'' definitions, e.g.\@ QIITs,
may in general be neither injective nor disjoint.
For instance, $\LUB (\lam n \bot,q)=\bot$.
However, we have the following lemma:
\begin{corollary} \label{cor:eta-order-y}
 For any $a:A$ and $y : \pbot A$, we have that $\eta(a) \order y$ implies that $\eta(a) = y$.
 Furthermore $\eta$ is injective: if $\eta(a) = \eta(b)$, then $a = b$.
 Moreover we have $\eta(a) \not= \bot$.
\end{corollary}
\begin{proof}[of Corollary~\ref{cor:eta-order-y}] 
 The last two claims are simple consequences of the lemma and reflexivity.
 For the first claim, let us fix $a : A$ and apply
 Lemma~\ref{lem:simplified-elim} with $P(y) \jdgeq \eta(a) \order y
 \to \eta(a) = y$.
 The only non-immediate step is the case for $\LUB(s,q)$, where we can assume $\prd{n : \N} P(s_n)$. 
 From $\eta(a) \order \LUB  (s,q)$ and Lemma~\ref{lem:order-lemma} we get $\trunc{\sm{n:\N}\eta(a) \order s_n}$.
 We are proving a proposition, so we can assume that we have $n:\N$ such that $\eta(a) \order s_n$.
 This implies that, for all $m \geq n$, $\eta(a) \order s_m$ and
 hence, by the ``inductive hypothesis'', $\eta(a) = s_m$.
 Thus $\eta(a)$ is an upper bound of~$s$, so we get $\LUB(s,q) \order
 \eta(a)$, which by antisymmetry implies $\LUB(s,q) = \eta(a)$.
 \qed
\end{proof}

The proof of the lemma is more technical.
The approach is similar to that used to prove some results about the
real numbers defined as a HIIT in the HoTT book~\cite[Theorems 11.3.16
and 11.3.32]{HoTTbook}.
We only give a sketch here, the complete proof can be found in our
Agda formalisation.
\begin{proof}[of Lemma~\ref{lem:order-lemma}]
  For every $a : A$ we construct a relation in $\pbot A \to
  \mathsf{Prop}$ by applying the elimination principle of $\pbot A$
  and $\order$, treating $\mathsf{Prop}$ as a partiality algebra over
  $A$ in the following way:
 \begin{alignat*}{6}
  &P \order_{\Prop} Q & \; & \jdgeq & \; &(P \to Q) &\qquad \qquad &    \eta_{\Prop}(b) & \; & \jdgeq & \; &(a = b) \\
  &\bot_{\Prop} && \jdgeq  && \zero && \LUB_{\Prop}(S,P) && \jdgeq  && \trunc{\sm{n:\N} S_n}
 \end{alignat*}
 Propositional extensionality is used to prove that $\mathsf{Prop}$ is
 a set (this is a variant of an instance of Theorem~7.1.11 in the HoTT
 book~\cite{HoTTbook}), and to prove the antisymmetry law.

Using Lemma~\ref{lem:simplified-elim} one can then show that the
defined relation is pointwise equal to $\eta(a) \order -$, and it is
easy to see that the relation has the properties claimed in the
statement of Lemma~\ref{lem:order-lemma}.
\qed
\end{proof}

Using the results above one can prove that the order is flat, in the
sense that if $x$ and $y$ are distinct from $\bot$ and $x \not= y$,
then $x \not\order y$ (see the formalisation).

\section{Relation to the Coinductive Construction} \label{sec:countchoice-equiv}

In this section we compare our QIIT to \citeauthor{capretta:2005}'s coinductive delay monad~\cite{capretta:2005}, quotiented by weak bisimilarity~\cite{Chapman2015}.
Let us start by giving \citeauthor{capretta:2005}'s construction,
as already outlined in the introduction.
$\Univ$ stands for a universe of types.%
\begin{definition}[delay monad and weak bisimilarity]
 For a set $A$ the \emph{\mbox{delay} monad} $\DD A$ is the coinductive type
 generated by $\now : A \to \DD A$ and $\later : \DD A \to \DD A$.
 The ``terminates with'' relation $\downarrow_{\DDinSubscript A}\ : \DD A \to A \to \Univ$ is the indexed inductive type generated by two constructors of type $\eta(a) \downarrow_{\DDinSubscript A} a$ and \mbox{$p \downarrow_{\DDinSubscript A} a \to \later(p) \downarrow_{\DDinSubscript A} a$}.
 Furthermore, $x$ and $y : \DD A$ are said to be \emph{weakly bisimilar}, written $x \sim_{\DDinSubscript A} y$, if $\prd{a:A} x \downarrow_{\DDinSubscript A} a \leftrightarrow y \downarrow_{\DDinSubscript A} a$ holds.
\end{definition}
It is easy to give $\DD A$ the structure of a monad.
Note that $x \downarrow_{\DDinSubscript A} a$ can alternatively be defined to be $\sm{n : \N} x = \later^n(\now (a))$.
The types $x \downarrow_{\DDinSubscript A} a$ and $x
\sim_{\DDinSubscript A} y$ are propositional, and
$\sim_{\DDinSubscript A}$ is an equivalence relation on $\DD A$.

The goal of this section is to show that, in the presence of countable choice, the partiality monad $\pbot A$ is equivalent to $\quotient{\DD A}{\sim_{\DDinSubscript A}}$.
(We use the notion of equivalence from the HoTT book~\cite{HoTTbook}, which for sets is equivalent to bijective correspondence.)
To understand the structure of the proof, let us observe that $\quotient{\DD A}{\sim_{\DDinSubscript A}}$ is constructed as a ``coinductive type that is quotiented afterwards'', while $\pbot A$ is an ``inductive type that is quotiented at the time of construction''.
To build a connection between these, it seems rather intuitive to consider an intermediate construction,
either a ``coinductive type that is quotiented at the time of construction'' or an ``inductive type that is quotiented afterwards''.
The theory of ``higher coinductive types'' has, as far as we know, not been explored much yet, so we go with the second option.
We do not even need an \emph{inductive} construction: it is well-known that coinductive structures can be represented using finite approximations, and here, it is enough to consider monotone functions.
Thus, first we will show that $\DD A$ is equivalent to a type of monotone sequences, carefully formulated, and that the equivalence lifts to the quotients.
Then we will prove that, assuming countable choice, the quotiented monotone sequences are equivalent to $\pbot A$.

\subsection{The Delay Monad and Monotone Sequences}

For a set $A$ we say that a function $g : \N \to A + \one$ is a monotone sequence if it satisfies the propositional property
\begin{equation*}
  \ismon(g) \; \jdgeq \; \prd{n : \N} \left(g_n = g_{n+1}\right) + \left((g_n = \inr(\star))   \times    (g_{n+1} \not= \inr(\star))\right).
\end{equation*}
The set of monotone sequences, $\sm{g : \N \to A + \one} \ismon(g)$,
is denoted by $\Seq_A$.
Below the notation $-_n$ will be used not only for functions, but also
for monotone sequences; $(g,p)_n$ means $g_n$.

As \citeauthor{Chapman2015}~\cite{Chapman2015} observe, one can construct a sequence of type $\N \to A + \one$ from an element of $\DD A$.
If their construction is tweaked a little, then the resulting sequences are monotone, and the map is an equivalence:
\begin{lemma} \label{lem:seq-coind-equiv}
 The types $\Seq_A$ and $\DD A$ are equivalent.
\end{lemma}
\begin{proof}
We can simply give functions back and forth.
Note that endofunctions on $\DD A$ that correspond to $\later$ and ``remove $\later$, if there is one'' can be mimicked for $\Seq_A$:
let us use the names $\shift$ and $\unshift : \Seq_A \to \Seq_A$ for the functions that are determined by $\shift(g)_0 \jdgeq \inr(\star)$, $\shift(g)_{n+1} \jdgeq g_n$, and $\unshift(g)_n \jdgeq g_{n+1}$.

Define $j : \DD A \to \Seq_A$ such that $j(\now (a))$ equals $\lam n \inl(a)$, and $j(\later(x))$ equals $\shift(j(x))$.
One way to do this is to define $j(z)_n$ by recursion on $n$, followed by case distinction on $z$.
Furthermore, define $h : \Seq_A \to \DD A$ in the following way:
Given $s : \Seq_A$, do case distinction on $s_0$.
If $s_0$ is $\inl(a)$, return $\now (a)$.
Otherwise, return $\later (h(\unshift(s)))$.
It is straightforward to show that $j$ and $h$ are inverses of each other.
\qed
\end{proof}
As an aside, our formalisation shows that
Lemma~\ref{lem:seq-coind-equiv} holds even if $A$ is not a set.

Next, we mimic the relation
$\downarrow_{\DDinSubscript{-}}$ by setting
\begin{align*}
 &{\downarrow_{\Seq{}}} : \Seq_A \to A \to \Univ \\
 &s \downarrow_{\Seq{}} a \; \jdgeq \; \sm{n:\N} s_n = \inl (a).
\end{align*}
 The relation $\downarrow_{\Seq{}}$ is not in general propositional. To remedy this, we can truncate and consider $\trunc{s \downarrow_{\Seq{}} a}$.
 Using strategies explained 
 by \citeauthor{krausgeneralizations}~\cite{krausgeneralizations},
 we have $\trunc{s \downarrow_{\Seq{}} a} \to s \downarrow_{\Seq{}} a$, so we can always extract a concrete value of $n$:
 a variant of the definition above in which the number $n$ is required to be minimal is propositional, and this definition can be shown to be logically equivalent to both $\trunc{s \downarrow_{\Seq{}} a}$ and $s \downarrow_{\Seq{}} a$. 
 See the formalisation for details.

We define the propositional relations $\order_{\Seq{}}$ and $\sim_{\Seq{}}$ by
\begin{align*}
 s \order_{\Seq{}} t \; &\jdgeq \; \prd{a:A} \trunc{s \downarrow_{\Seq{}} a} \to \trunc{t \downarrow_{\Seq{}} a}\text{ and}\\
 s \sim_{\Seq{}} t \; &\jdgeq \; s \order_{\Seq{}} t \times t \order_{\Seq{}} s\text{.}
\end{align*}
By checking that the equivalence from Lemma~\ref{lem:seq-coind-equiv} maps $\sim_{\Seq{}}$-related elements to $\sim_{\DDinSubscript A}$-related elements, we get:
\begin{lemma} \label{lem:quot-equiv}
 The sets $\quotient{\Seq_A}{\sim_{\Seq{}}}$ and $\quotient{\DD A}{\sim_{\DDinSubscript A}}$ are equivalent. \qed
\end{lemma}

\subsection{Monotone Sequences and the QIIT Construction} \label{subsec:SeqVsQIIT}

As the final step of showing that $\quotient{\DD A}{\sim_{\DDinSubscript A}}$ and $\pbot A$ are equivalent, we show that $\quotient{\Seq_A}{\sim_{\Seq{}}}$ and $\pbot A$ are.
The plan is as follows:
There is a canonical function $w : \Seq_A \to \pbot A$ which can be extended to a function 
$\tilde w : \quotient{\Seq_A}{\sim_{\Seq{}}} \to \pbot A$.
The function $\tilde w$ is injective.
Furthermore, if we assume countable choice, then the function $w$, and thus also $\tilde w$, are surjective.
Thus $\tilde w$ is an equivalence.

Let us start by constructing $w$ and $\tilde w$.
We use a copairing function $\caseta : A + \one \to \pbot A$ defined by
$\caseta (\inl(a)) \jdgeq \eta(a)$ and $\caseta  (\inr(\star)) \jdgeq \bot$, and
define $w : \Seq_A \to \pbot A$ by $w(s,q) \jdgeq \bigsqcup (\caseta  \circ s , \ldots)$, with a canonical proof of monotonicity.
\begin{lemma} \label{lem:w-tilde}
 The function $w$ is monotone:
 $\prd{s,t:\Seq_A}s \order_{\Seq{}} t \to w(s) \sqsubseteq w(t)$.
 Thus $w$ extends to a map $\tilde w : \quotient{\Seq_A}{\sim_{\Seq{}}} \to A_\bot$.
\end{lemma}
\begin{proof}
For the second claim we show that $w$ maps elements related by $\sim_{\Seq{}}$ to equal elements.
This follows from the first claim by antisymmetry.
 For the first claim it suffices to find a function $k : \N \to \N$ such that, for all $n$, we have $\caseta(s_n) \sqsubseteq \caseta(t_{k(n)})$.
 Fix $n$.
 If $s_n$ is $\inr(\star)$, then $\caseta(s_n)$ is $\bot$, and $k(n)$ can thus be chosen arbitrarily.
 If $s_n$ is $\inl(a)$, then we have $s \downarrow_{\Seq{}} a$ and therefore $t \downarrow_{\Seq{}} a$, which gives us a number $k(n)$ such that $t_{k(n)} = \inl (a)$ and $\caseta(s_n) = \eta(a) = \caseta(t_{k(n)})$.
 \qed
\end{proof}

\begin{lemma}
 The function $\tilde w$ is injective:
  $\prd{s,t : \quotient{\Seq_A}{\sim_{\Seq{}}}} \tilde w (s) = \tilde w (t) \to s = t$.
\end{lemma}
\begin{proof}
It suffices to show that, for $s,t : \Seq_A$, $w(s) = w(t)$ implies $s \sim_{\Seq{}} t$.
By symmetry, it is enough to fix $a:A$ and show $\trunc{s \downarrow_{\Seq{}} a} \to \trunc{t \downarrow_{\Seq{}} a}$, which follows from $s \downarrow_{\Seq{}} a \to \trunc{t \downarrow_{\Seq{}} a}$.
If $s \downarrow_{\Seq{}} a$ then $w(s) = \eta(a)$, and thus also $w(t) = \eta(a)$.
Using Lemma~\ref{lem:order-lemma} and Corollary~\ref{cor:eta-order-y} we then get $\trunc{\sm{n:\N} \eta(a) = \caseta(t_n)}$, which implies $\trunc{\sm{n:\N} t_n = \inl(a)}$.
\qed
\end{proof}

\begin{lemma}
 Under countable choice, $w$ is surjective:
  $\prd{x : A_\bot} \trunc{\sm{s : \Seq_A} w(s) = x}$.
\end{lemma}
\begin{proof}
 We apply the simplified induction principle presented in Lemma~\ref{lem:simplified-elim}.
 The propositional predicate is $P(x) \jdgeq \trunc{\sm{s : \Seq_A} w(s) = x}$.
 Both $P(\bot)$ and $P(\eta(a))$ are trivial: in the first case we use the sequence that is constantly $\inr(\star)$, while in the second case we take the one that is constantly $\inl(a)$.
 
 The interesting part is to show $P\left(\LUB(f,p)\right)$ for a given $f : \N \to \pbot A$ and $p : \prd{n : \N} f_n \order f_{n+1}$.
 By the mentioned induction principle, we can assume $\prd{n : \N} P(f_n)$, which unfolds 
 to
$\prd{n:\N} \trunc{\sm{t : \Seq_A} w(t) = f_n}$.
By countable choice, we can swap $\Pi_{n:\N}$ and $\trunc -$, which allows us to remove the truncation completely, because the goal is propositional.
Hence we can assume $\prd{n:\N} \sm{t : \Seq_A} w(t) = f_n$.

Using the usual distributivity law for $\prd{}$ and $\sm{}$ (sometimes called the ``type-theoretic axiom of choice''), we can assume that we are given $g : \N \to \Seq_A$ and a proof $\gamma : \prd{n:\N} w(g_n) = f_n$.
By dropping the monotonicity proof and uncurrying, $g$ gives us a function $g' : \N \times \N \to A + \one$ with the property that it assumes at most one value in $A$:
If $g'_{i,j} = \inl(a)$, then (using $\gamma$) $\eta(a) \order f_i$,
thus $\eta(a) \order \LUB(f,p)$, and hence $\LUB(f,p) = \eta(a)$ by Corollary~\ref{cor:eta-order-y}.
If we also have $g'_{k,m} = \inl(b)$, then $\eta(a) = \eta(b)$, which by Corollary~\ref{cor:eta-order-y} implies that $a = b$.

We use $g'$ to construct an element of $\Seq_A$.
Take an arbitrary isomorphism $\sigma : \N \to \N \times \N$ (a split surjection would also be sufficient),
and define a function $\tilde g : \N \to A + \one$ by
\begin{equation*}
  \tilde g(n) \; \jdgeq \; \begin{cases}
                            g'(\sigma_n)\text{,} & \mbox{if $n = 0$ or $g'(\sigma_n) \not= \inr(\star)$,} \\
                            \tilde g(n-1)\text{,} & \mbox{otherwise.}
                           \end{cases}
 \end{equation*}
The intuition is that $\tilde g(n)$ checks the first $n + 1$ results of $g' \circ \sigma$ and chooses the last which is of the form $\inl(-)$, if any, otherwise returning $\inr(\star)$.
Because $g'$ assumes at most one value in $A$ we get that $\tilde g$ is monotone, $q : \ismon(\tilde g)$.
Furthermore $(\tilde g, q) \downarrow_{\Seq{}} a$ holds if and only if we have $\sm{n:\N}g'(\sigma_n)=\inl(a)$.

In order to complete the proof of $P\left(\LUB(f,p)\right)$ we show that $w(\tilde g, q) = \LUB(f,p)$ by using antisymmetry:
\begin{itemize}
\item \textit{First part: $w(\tilde g, q) \order \LUB(f,p)$.}
After unfolding the definition of $w$ we see that it suffices to prove $\caseta (\tilde g_n) \order \LUB(f,p)$ for an arbitrary $n : \N$.
If $\tilde g_n$ is $\inr(\star)$, then this is trivial.
If $\tilde g_n$ is $\inl(a)$ for some $a:A$, then we can find a pair $(i,j)$ such that $g'_{i,j} = \inl(a)$.
Thus we get the following chain:
\begin{equation*}
 \caseta (\tilde g_n) \; = \; \caseta (g'_{i,j}) \; \order \; w(g_i) \; = \; f_i \; \order \; \LUB(f,p)
\end{equation*}
\item \textit{Second part: $\LUB(f,p) \order w(\tilde g, q)$.}
  Given $n : \N$, we show that $f_n \order w(\tilde g, q)$.
  By $\gamma_n$ we have $f_n = w(g_n)$.
  Thus it suffices to prove $w(g_n) \order w(\tilde g, q)$, which by Lemma~\ref{lem:w-tilde} follows if $g_n \order_{\Seq{}} (\tilde g, q)$.
  If $g_n(i) = \inl(a)$ for some $i$ and $a$, then we have $\inl(a) = g'_{n,i} = g'(\sigma(\sigma^{-1}_{n,i}))$, and thus $(\tilde g, q) \downarrow_{\Seq{}} a$.
  \qed
\end{itemize}
\end{proof}

This immediately shows that $\tilde w$ is surjective as well.
Putting the pieces together, we get the main result of this section:
\begin{theorem}
 In the presence of countable choice the map $\tilde w$ is an equivalence.
 Hence the three sets $\quotient{\DD A}{\sim_{\DDinSubscript A}}$, $\quotient{\Seq_A}{\sim_{\Seq{}}}$ and $A_\bot$ are equivalent.
\end{theorem}
\begin{proof}
 A function between sets is an equivalence exactly if it is surjective and injective.
 This is a special case of Theorem~4.6.3 in the HoTT book~\cite{HoTTbook}, which states that a function between arbitrary types is an equivalence if and only if it is surjective and an \emph{embedding}, which for sets is equivalent to being injective.
 \qed
\end{proof}

\section{Applications} \label{sec:applications}

The following examples show that our construction can be used in formalisations.

\subsection{Nonterminating Functions as Fixed Points}
\label{sec:nonterminating-functions}

Partiality algebras can be used to implement not necessarily
terminating functions. Let $(Y, \order_Y, \bot_Y, \eta_Y, \LUB_Y)$ be
a partiality algebra, and let $\varphi : Y \to Y$ be a monotone and
$\omega$-continuous function.
We can write down the least fixed point of $\varphi$ directly as
$\LUB_Y (\lam n \varphi^n(\bot_Y) , p)$, where $p$ is constructed from
the fact that $\bot_Y \order_Y \varphi(\bot_Y)$ and from the
monotonicity proof of $\varphi$. One does not need $\omega$-continuity
to write down this expression, but we use it to prove that the
expression is a fixed point of $\varphi$.
 
If $(Y, \order_Y, \bot_Y, \eta_Y, \LUB_Y)$ is a partiality algebra and
$X$ is any type, then the function space $X \to Y$ can be given the
structure of a partiality algebra in a canonical way (this is done for
\emph{dependent} types $\prd{x:X}Y(x)$ in the formalisation).
As an example of how this kind of partiality algebra can be used we
will construct a function $\search_q: A^\omega \to \pbot A$ that takes
an element of the coinductive set of streams $A^\omega$ and searches
for an element of the set $A$ satisfying the decidable predicate $q :
A \to \two$.
The function is constructed as the least fixed point of the following
endofunction on $A^\omega \to \pbot A$:
 \begin{align*}
  &\Phi(f)(a \dblcolon \mathit{as}) \; \jdgeq \; \mathsf{if}\ q(a)\ \mathsf{then}\ \eta(a)\ \mathsf{else}\ f(\mathit{as})
 \end{align*}
 It is straightforward to check that $f \order g$ implies $\Phi(f)
 \order \Phi(g)$ by applying $\Phi(f)$ and $\Phi(g)$ to a point $a
 \dblcolon \mathit{as}$ and doing case analysis on $q(a)$.
Thus $\Phi$ is monotone.
In a similar way one can verify that $\Phi$ is $\omega$-continuous.

\subsection{Functions from the Reals} \label{subsec:reals}

\newcommand{\R}{\mathbb{R}^q}

Let us consider the Cauchy reals, defined as a quotient.
We say that $f : \N \to \Q$ is a \emph{Cauchy sequence} if, for all $m$, $n : \N$ with $m < n$, we have $-1 < m \cdot (f_m - f_n)$ $< 1$.
Furthermore, $f$ and $g$ are equivalent (written $f \sim g$) if, for all $n : \N$, we have $-2 \leq$ $n \cdot(f_n - g_n) \leq 2$.
We use the notation $\R$ for the quotient of Cauchy sequences by~$\sim$.

A meta-theoretic result is that, without further assumptions, any definable function (i.e.\@ any closed term) of type $\R \to \two$ is constant
for reasons of
continuity~\cite{nuo:thesis}.
In particular, we cannot define a function $\mathit{isPositive}$ which checks whether a real number is positive.
However, we \emph{can} define a function $\mathit{isPositive} : \R \to
\pbot{\two}$ such that $\mathit{isPositive}(r)$ is equal---but not
judgmentally/definitionally equal---to $\eta(\true)$ if $r$ is
positive, $\eta(\false)$ if $r$ is negative, and $\bot$ if $r$ is
zero.

We define this function as follows:
Given a Cauchy sequence $f : \N \to \Q$, we construct a new sequence
$\overline f : \N \to \{-, ?, +\}$.
The idea is that $\overline f_n$ is an approximation which only takes $f_i$ with $i \leq n$ into account.
We start with $\overline f_0 \jdgeq \;?$.
If we have chosen $\overline f_{n-1}$ to be $-$, then we choose $\overline f_n$ to be $-$ as well, and analogously for $+$.
If we have chosen $\overline f_{n-1}$ to be $?$, we check whether $f_n \cdot n < -2$, in which case
we choose $\overline f_n$ to be $-$; if $f_n \cdot n > 2$,
we choose $\overline f_n$ to be $+$; otherwise, we choose
$\overline f_n$ to be $?$.
We can compose with the map $\{-,?,+\} \to \pbot A$ which is defined by $- \mapsto \eta(\false)$, $? \mapsto \bot$, and $+ \mapsto \eta(\true)$.
This defines a monotone sequence in $\pbot \two$, and we can form $\LUB \overline f : \pbot \two$ to answer whether $f$ represents a positive or negative number, or is zero. 
One can check that equivalent Cauchy sequences get mapped to equal values, hence we get $\mathit{isPositive} : \R \to \pbot \two$.

The strategy outlined above does not quite work for the reals defined
as a HIIT~\cite{HoTTbook} because, roughly speaking, in that setting
$f_n$ is a real number and a comparison such as $f_n \cdot n < -2$ is
undecidable.
Recently \citeauthor{gilbert_reals} has refined our approach and
defined a function
$\mathit{isPositive}$ for such reals~\cite{gilbert_reals}, using the
definition of the partiality monad presented in this text (with
insignificant differences).
\citeauthor{gilbert_reals}'s key observation is that comparisons
between real numbers and rational numbers are semidecidable, and
semidecidability is sufficient to define $\mathit{isPositive}$.

\subsection{Operational Semantics} \label{subsec:operational-semantics}

In previous work the second-named author has discussed how one can use
the delay monad to express operational semantics as definitional
interpreters~\cite{ICFP-2012-Danielsson}.
As a case study we have ported some parts of this work to the
partiality monad discussed in the present text:
definitional interpreters for a simple functional language and a
simple virtual machine, a type soundness proof, a compiler, and a
compiler correctness result.
Due to lack of space we do not include any details here, but refer
interested readers to the accompanying source code.

\section{Discussion and Further Work} \label{sec:discussions}

We have constructed a partiality monad without using countable choice.
This is only a first step in the development of a form of constructive
domain theory in type theory.
It remains to be seen whether it is possible to, for instance,
replicate the work of
\citeauthor{benton-kennedy-varming}~\cite{benton-kennedy-varming}, who
develop domain theory using the delay monad.

Consider the partial function
$\mathit{filter} : \prd{A : \Set} (A \to \two) \to A^\omega \to
A^\omega$
that filters out elements from a stream.
How should partial streams over $A$ be defined?
Defining them as $\nu X. \pbot{(A \times X)}$ seems inadequate, because
the ordering of $\pbot{(-)}$ is flat.
One approach would perhaps be to define this type by solving a domain
equation.
Instead of relying on the type-theoretic mechanism to define recursive
types, we can perhaps construct a suitable type of partial streams as
the colimit of an $\omega$-cocontinuous functor on the category of
$\omega$-cpos.
Preliminary investigations indicate that QIITs are useful in the
endeavour, for example in the definition of a lifting comonad on
$\omega$-cpos (as suggested by Paolo Capriotti).

Going in another direction, it might be worth investigating how much
topology can be done using the Sierpinski space, represented as
$\pbot \one$ in our setting.
A very similar question was discussed at the Special Year on Univalent
Foundations of Mathematics at the IAS in Princeton
(2012--2013).
Moreover, some observations have been presented by
\citeauthor{gilbert_reals}~\cite{gilbert_reals}, who used our
definition of $\pbot \one$ as presented in this paper (with minor
differences).

\subsubsection*{Acknowledgements.}
We thank Gershom Bazerman, Paolo Capriotti, Bernhard Reus, and Bas
Spitters for interesting discussions and pointers to related work, and
the anonymous reviewers for useful feedback.
The work presented in Sect.~\ref{subsec:characterisation-of-order} was
done in collaboration with Paolo Capriotti.

\bibliographystyle{plainnat}
\bibliography{partiality}
\end{document}